\documentclass[a4paper,10pt]{scrartcl}

\usepackage[utf8]{inputenc}
\usepackage[T1]{fontenc}
\usepackage[english]{babel}

\usepackage{amsmath,amssymb,amsfonts,amsthm}

\usepackage[ruled, vlined, linesnumbered]{algorithm2e}
\usepackage{xspace}

\usepackage{fancyhdr}

\usepackage[margin=3cm]{geometry}
\newcommand{\R}{\mathbb{R}}

\newcommand{\N}{\mathbb{N}}
\newcommand{\cO}{\mathcal{O}}
\newcommand{\cR}{\mathcal{R}}
\newcommand{\tr}{\tilde r}
\newcommand{\tO}{\tilde O}

\newcommand{\abs}[1]{\left\vert #1 \right\vert}
\newcommand{\norm}[1]{\left\Vert #1 \right\Vert}
\newcommand{\OneTo}[2]{#1\in[#2]}
\newcommand{\one}{\pmb{1}}
\newcommand{\wmax}{w_{max}^B}
\newcommand{\cif}{& \text{if }}

\newcommand{\SAT}{\textsc{Max 2-Sat}\xspace}
\newcommand{\SATF}{{\normalfont SAT/Flip}\xspace}
\newcommand{\KM}{\textsc{Discrete $K$-Means}\xspace}
\newcommand{\KMS}{{\normalfont DKM/Swap}\xspace}
\newcommand{\UFL}{\textsc{Uncapacitated Facility Location}\xspace}
\newcommand{\MUFL}{\textsc{Metric Uncapacitated Facility Location}\xspace}
\newcommand{\MUFLS}{{\normalfont MUFL/Swap}\xspace}
\newcommand{\NAE}{\textsc{Pos NAE Max 2-Sat}\xspace}
\newcommand{\NAEF}{{\normalfont PNAESAT/Flip}\xspace}
\newcommand{\FKM}{\textsc{Discrete Fuzzy $K$-Means}\xspace}
\newcommand{\FKMS}{{\normalfont DFKM/Swap}\xspace}

\newcommand{\pkm}{\phi_{KM}}
\newcommand{\pfl}{\phi_{FL}}
\newcommand{\pfkm}{\phi_{FKM}}

\newtheorem{definition}{Definition}

\newtheorem{theorem}[definition]{Theorem}
\newtheorem{lemma}[definition]{Lemma}
\newtheorem{claim}[definition]{Claim}
\newtheorem{corollary}[definition]{Corollary}

\newtheorem{proposition}[definition]{Proposition}

\setkomafont{paragraph}{\normalfont\itshape\large}

\fancypagestyle{plain}{
\fancyhf{}
\cfoot{\normalfont \hfill \pagemark \hfill}

}
\allowdisplaybreaks

\makeatletter
\newcommand*{\defeq}{\mathrel{\rlap{%
                     \raisebox{0.3ex}{$\m@th\cdot$}}%
                     \raisebox{-0.3ex}{$\m@th\cdot$}}%
                     =}
\makeatother

\begin{document}

\thispagestyle{empty}

\begin{center}
	{\LARGE \textbf{Complexity of Single-Swap Heuristics for Metric Facility Location and Related Problems}\footnote{This is a full version of the paper with the same name that will be presented at CIAC 2017.}}
	\bigskip
	
	{\Large Sascha Brauer}
	
	\texttt{sascha.brauer{@}uni-paderborn.de}
	\bigskip
	
	Department of Computer Science
	
	Paderborn University
	
	33098 Paderborn, Germany
\end{center}  
\begin{abstract}
% 	$K$-means is one of the most widely used clustering objectives.
% 	Its most used local search heuristic, the $K$-means algorithm, is known to have exponential worst-case running time and to produce arbitrarily bad solutions.
% 	Kanungo et al. proved that an alternative, swap-based, local search approach yields an $\cO(1)$ approximation to the $K$-means problem.
% 	In this paper, we show that this local search heuristic applied to the discrete $K$-means problem is tightly PLS-complete and hence has exponential worst-case running time, as well.
% 	We further extend this result to the discrete fuzzy $K$-means problem.
	Metric facility location and $K$-means are well-known problems of combinatorial optimization.
	Both admit a fairly simple heuristic called \emph{single-swap}, which adds, drops or swaps open facilities until it reaches a local optimum.
	For both problems, it is known that this algorithm produces a solution that is at most a constant factor worse than the respective global optimum.
	In this paper, we show that single-swap applied to the weighted metric uncapacitated facility location and weighted discrete $K$-means problem is tightly PLS-complete and hence has exponential worst-case running time.
\end{abstract}

\section{Introduction}

Facility location is an important optimization problem in operations research and computational geometry.
Generally speaking, the goal is to choose a set of locations, called facilities, minimizing the cost of serving a given set of clients.
The service cost of a client is usually measured in some form of distance from the client to its nearest open facility.
To prevent the trivial solution of opening a facility at each possible location, we usually introduce some sort of opening cost penalizing large sets of open facilities.
This general framework comprises a plethora of problems using different functions to measure distance and combinations of opening and service cost.
In this paper, we discuss two popular problems closely related to facility location: \MUFL (MUFL) and \KM (DKM).

\subsection{Problem Definitions}

In an \UFL (UFL) problem we are given a set of clients $C$, a weight function $w:C\rightarrow\N$ on the clients, a set of facilities $F$, an opening cost function $f:F\rightarrow\R$, and a distance function $d: C\times F\rightarrow \R$.
The goal is to find a subset of facilities $O\subset F$ minimizing

\[ \pfl(C,F,O) = \sum_{c\in C} w(c)\min_{o\in O}\{d(c,o)\} + \sum_{o\in O} f(o) \ . \]

This problem is uncapacitated in the sense, that any open facility can serve, i.e. be the nearest open facility to, any number of clients.
Simply speaking, opening a lot of facilities incurs high opening cost, but small service cost, and vice versa.
MUFL is a special case of this problem, where we require the distance function $d$ to be a metric on $C\cup F$.

% DKM is a variant of UFL, where we assume $C = F \subset \R^D$, a constant opening cost $f \equiv 0$, and measure the distance between clients as $d(p,q) = \norm{p-q}^2$.
% To make this problem non-trivial, we only allow solutions opening at most a fixed number of $K$ facilities.
DKM is a problem closely related to UFL, where we do not differentiate between clients and facilities, but are given a single set of points $C\subset\R^D$.
We measure distance between points $p,q\in C$ as $d(p,q) = \norm{p-q}^2$.
Furthermore, instead of imposing an opening cost, we allow at most $K$ locations to be opened.
Hence, the goal is to find $O\subset C$ with $\abs{O} = K$ minimizing

\[ \pkm(C,O) = \sum_{c\in C} w(c) \min_{o\in O}\{\norm{c - o}^2 \} \ . \]

Notice, that we consider the \emph{weighted} variant of both MUFL and DKM, where each client is associated with a positive weight.
Such a weight can be interpreted as the importance of serving the client or as multiple clients present in the same location.

\subsection{Local Search}

A popular approach to solving hard problems of combinatorial optimization is local search.
The general idea of a local search algorithm is to define a small \emph{neighbourhood} for each feasible solution.
Given a problem instance and an initial solution, the algorithm replaces the current solution by a better solution from its neighbourhood.
This is repeated until the algorithm finds no improvement, hence has found a solution that is not worse than any solution in its neighbourhood.
The runtime and the quality of the produced solutions of a local search algorithm depend heavily on its definition of neighbourhood.

Theoretical aspects of local search are captured in the definition of the complexity class PLS. 
There is a special type of reduction, called PLS-reduction, with respect to which PLS has complete problems \cite{johnson88}.
Notably, there are PLS-complete problems, which exhibit two important properties.
First, given an instance and an initial solution, it is PSPACE-complete to find a locally optimal solution computed by a local search started with the given initialization.
Second, there is an instance and an initial solution, such that this initial solution is exponentially many local search steps away from every locally optimal solution \cite{monien10}.
There is a stronger version of PLS-reductions, so-called \emph{tight} PLS-reductions which are of special interest, as they preserve both of these properties \cite{papadimitriou90}.
PLS-complete problems having these two properties are therefore sometimes called \emph{tightly} PLS-complete.

In the following, we examine a local search algorithm for MUFL und DKM called the \emph{single-swap heuristic}.
For MUFL, we allow the algorithm to either close an open facility, newly open a closed facility or do both in one step (\emph{swap} an open facility).
Since feasible DKM solutions consist of exactly $K$ open facilities, we do not allow the algorithm to solely open or close a facility, but only to swap open facilities.
Formally, we define these respective neighbourhoods as
\begin{align*}
	N_{MUFL}(O) &= \{O'\subset F \;|\; \abs{O \setminus O'}\leq 1 \wedge \abs{O'\setminus O} \leq 1 \}\quad\text{and} \\
	N_{DKM}(O) &= \{O'\subset C \;|\; \abs{O \setminus O'} = 1 \wedge \abs{O'\setminus O} = 1 \}  \ . 
\end{align*}
By \MUFLS and \KMS we denote the respective problem as a PLS-problem associated with the described single-swap neighbourhood.

\subsection{Related Work}

Approximating MUFL has been subject to considerable amount of research using different algorithmic techniques.
The problem can be $4$-approximated using LP-rounding \cite{shmoys97}, $3$-approximated using a Primal-Dual technique \cite{jain01}, and $1.61$-approximated using a greedy algorithm \cite{jain02}.
However, it is known that there is no polynomial time algorithm approximating MUFL better than $1.463$ unless NP~$\subseteq$~DTIME$(n^{\log\log n})$ \cite{guha99}.
Arya et al. showed that the standard local search algorithm of \MUFLS computes a $3$-approximation for MUFL \cite{arya04}.

A popular generalization of DKM called $K$-means admits facilities to be opened anywhere in the $\R^D$ instead of restricting possible locations to the locations of the clients.
The most popular local search algorithm for the $K$-means problem is called $K$-means algorithm, or Lloyd's algorithm \cite{lloyd82}.
It is well-known that the solutions produced by the $K$-means algorithm can be arbitrarily bad in comparison to an optimal solution.
Furthermore, it was shown that in the worst case, the $K$-means algorithm requires exponentially many improvement steps to reach a local optimum, even if $D=2$ \cite{vattani11}.
Recently, Roughgarden and Wang proved that, given a $K$-means instance and an initial solution, it is PSPACE-complete to determine the local optimum computed by the $K$-means algorithm started on the given initial solution \cite{roughgarden16}.
This is in line with several papers proving the same result for the simplex method using different pivoting rules \cite{adler2014,fearnley15}.
Kanungo et al. proved that the standard local search algorithm of \KMS computes an $\cO(1)$-approximation for DKM and hence also for general $K$-means \cite{kanungo04}.
They argue that a variation of the single-swap neighbourhood, where we impose some lower bound on the improvement of a single step, yields an algorithm with polynomial runtime but a slightly worse approximation ratio.
However, there is no known upper bound on the runtime of the exact single-swap heuristic, even for unweighted point sets.
Another variation of single-swap is the multi-swap heuristic, where we allow the algorithm to simultaneously swap more than one facility in each iteration.
For a large enough neighbourhood, i.e. swapping enough facilities in a single iteration, this heuristic yields a PTAS in Euclidean space with fixed dimension \cite{cohen16} and in metric spaces with bounded doubling dimension \cite{friggstad16}.

\subsection{Our Contribution}

In this paper, we analyze the PLS complexity of \MUFLS and \KMS.
By presenting a tight reduction from \SAT, we show that both problems are tightly PLS-complete, hence that both local search algorithms require exponentially many steps in the worst case and that given some initial solution it is PSPACE-complete to find the solution computed by the respective algorithm started on this initial set of open facilities.
Our reduction only works for the, previously introduced, weighted variants of MUFL and DKM.
That is, we construct instances with a non-trivial weight for each client.
% It is easy to check that the approximation results for the single-swap heuristic \cite{arya04,kanungo04} also hold for the weighted version of both problems.
Furthermore, our reduction for DKM requires the dimension of the point set to be on the order of the number of points.
% However, the runtime of the heuristic only depends linearly on the weights and dimensionality, since these quantities only occur in the evaluation of the objective function.
The performance of the single-swap heuristic is basically unaffected from using the more general variants of MUFL and DKM, since the known approximation bounds also hold for the weighted version of both problems, and since the runtime of the heuristic only depends linearly on the weights and the dimension.
However, this means that our reduction is weaker than a proof of the same properties for the unweighted variants or for a constant number of dimensions would be.

\begin{theorem}\label{thm:main}
	\MUFLS and \KMS are tightly PLS-complete.
\end{theorem}

We prove the two parts of Theorem~\ref{thm:main} in Sections~\ref{sec:mufl} and \ref{sec:km}.

% \subsection{Overview}
% 
% The following part of the paper is divided into three parts.
% In Section~\ref{sec:prelim} we present the well-known \SAT problem, which serves as the basis for our reduction.
% In Section~\ref{sec:mufl} we prove that \MUFLS is tightly PLS-complete.
% We extend this result in Section~\ref{sec:km} by showing how to modify the previously presented reduction to obtain
% the same result for \KMS.

% Afterwards, we introduce the \FKM (DFKM) problem and argue that further modification of our reduction leads to the PLS-completeness of the single-swap heuristic for this problem, as well.
% We extend this result in Section~\ref{sec:fuzzy} by showing how to modify the previously presented reduction to obtain the same result with respect to the discrete fuzzy $K$-means problem with $m = 2$.
\section{Preliminaries}\label{sec:prelim}

% We present two variants of the classic satisfiability problem, which are elementary problems in the study of the class PLS.
We present \SAT (SAT), a variant of the classic satisfiability problem, which is elementary in the study of PLS.
An instance of SAT is a Boolean formula in conjunctive normal form, where each clause consists of exactly $2$ literals and has some positive integer weight assigned to it.
The cost of a truth assignment is the sum of the weights of all satisfied clauses.
% Similarly, an instance of \NAE is a weighted Boolean formula in conjunctive normal form and it uses the same notion of neighbourhood.
% However, each clause consists of $2$ \emph{positive} literals, and the cost of a truth assignment is given by the sum over the weights of all clauses whose two variables differ in their truth value.
The PLS problem \SATF consists of SAT, where the neighbourhood of an assignment is given by all assignments obtained by changing the truth value of a single variable.
% \SATF is tightly PLS-complete \cite{schaeffer91}.

\begin{theorem}[\cite{schaeffer91}]
% 	\SAT and \NAE are tightly PLS-complete.
	\SATF is tightly PLS-complete.
\end{theorem}

% \subsection{Notation}

% Let $B$ be a clause set over the variables $\{x_n\}_{\OneTo{n}{N}}$ and $T$ be a truth assignment.
For each clause set $B$ and truth assignment $T$ we denote the SAT cost of $T$ with respect to $B$ by $w(B,T)$. %, and by $w_{NAE}(B,T)$ the \NAE cost of $T$.
For a literal $x$ we denote the set of all clauses in $B$ containing $x$ by $B(x)$.
Further, we denote the set of all clauses in $B$ satisfied by $T$ by $B_t(T)$ and let $B_f(T) = B\setminus B_t(T)$.
Finally, we set $\wmax = \max_{b\in B}\{ w(b) \}$.
% Finally, we let $(T)_n$ be the truth value of the variable $x_n$ in $T$.
% In the following, \KM and \FKM denotes the respective problem with the following neighbourhood:
% The neighbourhood of a set $C\subset P$, $\abs{C} = K$ consists of all size $K$ subsets of $P$ sharing $K-1$ points with $C$.
% Observe, that Algorithm~\ref{alg:kanungo} is the standard algorithm of \KM and \FKM as PLS problems.
\section{The Facility Location Reduction}\label{sec:mufl}
In the following, we formulate and prove one of our main results.

\begin{proposition}\label{prop:mufl}
	\SATF $\leq_{PLS}$ \MUFLS and this reduction is tight.
\end{proposition}

The following proof of Proposition~\ref{prop:mufl} is divided into three parts.
First, we present our construction of a PLS-reduction $(\Phi,\Psi)$, second, we argue on the correctness of this reduction and finally we show that the reduction is tight.

\subsection[Construction of Phi and Psi]{Construction of $\Phi$ and $\Psi$}\label{subsec:construction}

First, we construct the function $\Phi$ mapping an instance $(B,w)\in$ \SAT over the variables $\{x_n\}_{\OneTo{n}{N}}$ to an instance $(C,\omega,F,f,d)\in$ \MUFL.
In the following, we denote $M \defeq \abs{B}$.
Each variable $x_n$ appears as a facility twice, once as a positive and once as a negative literal.
Formally, we set $F = \{x_n,\bar x_n\}_{\OneTo{n}{N}}$.
We further locate a client at each facility and a client corresponding to each clause, so $C = F \cup B$.
% We choose $F = \{x_n,\bar x_n\}_{\OneTo{n}{N}}$ and $C = B \cup F$.
We set the distance function $d: C\cup F\times C\cup F \rightarrow \R$ to
\[ d(p,q) = d(q,p) = 
	\begin{cases}
		0 \cif p = q \\
		1 \cif p = x_n \wedge q = \bar x_n \\
		\frac{4}{3} \cif (p = x_n \vee p = \bar x_n) \wedge q = b_m \wedge p\in b_m \\
		\frac{5}{3} \cif (p = x_n \vee p = \bar x_n) \wedge q = b_m \wedge \bar p\in b_m \\
		2 & \text{else.}
	\end{cases}
\]

Simply speaking, a literal has distance $1$ from its negation, clauses have distance $4/3$ from literals they contain, distance $5/3$ from literals whose negation they contain, and all other clients/facilities have distance $2$ from each other.
It is easy to see that $d$ is a metric.
The weight of a client corresponding to a clause is the same as the weight of the clause.
If a client corresponds to a literal, then its weight is $W = M\cdot\wmax$.
\[ \omega(p) = 
	\begin{cases}
		w(b_m) \cif p = b_m \\
		W & \text{else}
	\end{cases}
\]
The opening cost function is constant $f\equiv 2W$.

Second, we construct the function $\Psi$ mapping solutions of $\Phi(B,w)$ back to solutions of $(B,w)$.
Given a set $O\subset F$ we let each variable $x_n$ be true if the facility $x_n\in O$ and let it be false otherwise.

In the following, we denote $\Phi(B,w) = (C,\omega,F,2W,d)$, $\Psi(B,w,O) = T_O$, and $d(c,O) = \min_{o\in O}\{d(c,o)\}$.

\subsection[(Phi,Psi) is a PLS-reduction]{$(\Phi,\Psi)$ is a PLS-reduction}

To prove that $(\Phi,\Psi)$ is a PLS-reduction we need to argue that $T_O$ is locally optimal for $(B,w)$ if $O$ is locally optimal for $\Phi(B,w)$.
Observe, that $\Psi$ is not injective, since $\Phi(B,w)$ has more feasible solutions than $(B,w)$.
We can tackle this problem by characterizing a subset of solutions for $\Phi(B,w)$ we call \emph{reasonable} solutions.

\begin{definition}
	Let $O\subset F$. 
	We call $O$ \emph{reasonable} if $\abs{O} = N$ and 
	\[ \forall \OneTo{n}{N}: x_n\in O \vee \bar x_n\in O \ . \]
\end{definition}

Reasonable solutions have several useful properties, which we prove in the following.
The restriction of $\Psi$ to reasonable solutions is a bijection, the MUFL cost of a reasonable solution is closely related to the SAT cost of its image under $\Psi$, and all locally optimal solutions of $\Phi(B,w)$ are reasonable.
This characterization of solutions is crucial to proving correctness and tightness of our reduction.

\begin{lemma}\label{lem:muflcost}
	If $O,O'\subset F$ are reasonable solutions for $\Phi(B,w)$, then
	\[ w(B,T_O) < w(B, T_{O'}) \Leftrightarrow \pfl(C,F,O) > \pfl(C,F,O') \ . \]
\end{lemma}

% We obtain Lemma~\ref{lem:muflcost} from the fact that the MUFL cost of a reasonable solution $O$ are a scaled variant of the SAT cost of $T_O$.
% Formally, we have that 
% 
% \[ \pfl(C,F,O) = \frac{4}{3} \sum_{b_m\in B_t(T_O)} w(b_m) + \frac{5}{3} \sum_{b_m\in B_f(T_O)} w(b_m) + 3WN \ . \]
% 
% A detailed proof of this claim can be found in the appendix.
\begin{proof}
	The following claim is essential to our proof of Lemma~\ref{lem:muflcost}.
	
	\begin{claim}\label{claim:appmuflcost}
		If $O\subset F$ is reasonable, then
		\[ \pfl(C,F,O) = \frac{4}{3} \sum_{b_m\in B_t(T_O)} w(b_m) + \frac{5}{3} \sum_{b_m\in B_f(T_O)} w(b_m) + 3WN \ . \]
	\end{claim}

	\begin{proof}
		Since $\abs{O} = N$, we have that the total opening cost of facilities is $2WN$.
		Observe, that since either $x_n$ or $\bar x_n$ is in $O$, we further obtain that the total service cost of all clients $\{x_n,\bar x_n\}_{\OneTo{n}{N}}$ is $WN$.
		Similar to before, we can observe a one-to-one mapping of the truth assignment of a variable, to whether the corresponding positive or negative literal is in $O$.
		It is easy to see that the clients corresponding to a clause in $B_t(T_O)$ have at least one open facility at distance $4/3$, while the clients corresponding to a clause in $B_f(T_O)$ have two facilities at distance $5/3$ and the rest at distance $2$.
	\end{proof}
	For the sake of brevity we introduce the notation 
	\[ B_{ab} = B_a(T_O) \cap B_b(T_{O'}) \ \]
	for $a,b\in \{t,f\}$.
	Observe that 
	\begin{align}
		w(B,T_O) < w(B,T_{O'}) \Leftrightarrow \sum_{b_m\in B_{tf}} w(b_m) < \sum_{b_m \in B_{ft}} w(b_m) \ . \label{eq:localopt}
	\end{align}
	Hence, using Lemma~\ref{claim:appmuflcost} we obtain
	\begin{align*}
		\pfl(C,F,O')
		 = 3WN + &\frac{4}{3}\sum_{b_m\in B_t(T_{O'})} w(b_m) + \frac{5}{3}\sum_{b_m\in B_f(T_{O'})} w(b_m) \\
		 = 3WN + &\frac{4}{3}\sum_{b_m\in B_{tt}} w(b_m) + \frac{4}{3}\sum_{b_m\in B_{ft}} w(b_m) +\\
		 &\frac{5}{3}\sum_{b_m\in B_{tf}} w(b_m) + \frac{5}{3}\sum_{b_m\in B_{ff}} w(b_m)\\
		 \overset{\eqref{eq:localopt}}{<} 3WN + &\frac{4}{3}\sum_{b_m\in B_{tt}} w(b_m) + \frac{5}{3}\sum_{b_m\in B_{ft}} w(b_m) +\\
		 &\frac{4}{3}\sum_{b_m\in B_{tf}} w(b_m) + \frac{5}{3}\sum_{b_m\in B_{ff}} w(b_m)\\
		 = 3WN + &\frac{4}{3}\sum_{b_m\in B_t(T_O)}w(b_m) + \sum_{b_m\in B_f(T_O)}w(b_m) 
		 = \pfl(C,F,O) 
	\end{align*}
\end{proof}

\begin{lemma}\label{lem:mufllocalproperty}
	If $O\subset F$ is a locally optimal solution for $\Phi(B,w)$, then $O$ is reasonable.
\end{lemma}

A detailed proof of Lemma~\ref{lem:mufllocalproperty} can be found in Section~\ref{subsec:muflproof}.
We can combine these results to obtain the correctness of our reduction.

\begin{corollary}\label{cor:correctness}
	If $O$ is locally optimal for $\Phi(B,w)$, then $T_O$ is locally optimal for $(B,w)$.
\end{corollary}

\begin{proof}
	Assume to the contrary that $T_O$ is not locally optimal.
	If $O$ is not reasonable, then it is not locally optimal by Lemma~\ref{lem:mufllocalproperty}.
	Therefore, assume that $O$ is reasonable.
	Since $T_O$ is not locally optimal, we know that there exists an $\OneTo{n}{N}$, such that $w(B,T_O^{\bar n}) > w(B,T_O)$, where $T_O^{\bar n}$ denotes $T_O$ with an inverted assignment of the $n^{th}$ variable.
	Since $O^{\bar n} \defeq (O\setminus\{x_n\})\cup\{\bar x_n\}$ is reasonable, $\Psi(B,w,O^{\bar n}) = T_O^{\bar n}$ and by Lemma~\ref{lem:muflcost} we know that
	\[ \pfl(C,F,O^{\bar n}) < \pfl(C,F,O) \ , \]
	and hence can conclude that $O$ is not locally optimal.
\end{proof}

\subsection{Proof of Lemma~\ref{lem:mufllocalproperty}}\label{subsec:muflproof}

The following proof of Lemma~\ref{lem:mufllocalproperty} is presented in two steps.
First, we argue in Lemma~\ref{lem:both} that no locally optimal solution can contain both a literal and its negation.
Second, we show in Lemma~\ref{lem:none} that every locally optimal solution contains a facility corresponding to each of the variables.
Combining these two results gives us Lemma~\ref{lem:mufllocalproperty} as a corollary.
From the following results we can moreover conclude that once the single-swap algorithm has reached a reasonable solution, it will always stay at a reasonable solution.
We take up on this fact in Section~\ref{subsec:tightness}, where we argue on the tightness of our reduction.

\begin{lemma}\label{lem:both}
	If $x_n,\bar x_n\in O$, then $O$ is not locally optimal.
\end{lemma}

\begin{proof}
	We show that closing the facility located at $x_n$ strictly decreases the cost, and hence that $O$ can not be locally optimal.
	When closing the facility $x_n$, we have to let all clients previously served by this facility (including the client located at $x_n$) be served by another facility.
	Choosing $\bar x_n$ as the replacement, we do not increase the cost by too much.
	More specifically, we can pay the additional cost with the cost we save from not opening $x_n$.
	Recall, that $B(x_n)$ is the set of all clauses containing the literal $x_n$, hence that $\abs{B(x_n)} \leq M$.
	Observe, that no client in $C\setminus B(x_n)$ (except $x_n$) is closer to $x_n$ than it is to $\bar x_n$.
	We obtain
	\begin{align*}
		\pfl(C,F,O) 
		&= \sum_{\substack{c\in C\setminus B(x_n)\\ c\neq x_n}}\omega(c) d(c,O) + \sum_{b_m\in B(x_n)} \omega(b_m) \frac{4}{3} + \abs{O}2W \\
		&> \sum_{\substack{c\in C\setminus B(x_n)\\ c\neq x_n}}\omega(c) d(c,O) + \sum_{b_m\in B(x_n)} \omega(b_m) \frac{5}{3} + W + (\abs{O}-1)2W \\
		&\geq \pfl(C,F,O\setminus\{x_n\}) \ .
	\end{align*}
\end{proof}

\begin{lemma}\label{lem:none}
	If $x_n,\bar x_n \not\in O$, then $O$ is not locally optimal.
\end{lemma}

\begin{proof}
	Similar to before, we show that opening a facility at $x_n$ strictly decreases the cost.
	When opening the facility at $x_n$ we have to save enough service cost by serving locations from it, that we can pay for opening the facility.
	Connecting the clients located at $x_n$ and $\bar x_n$ to the newly opened facility is sufficient.
	We obtain
	\begin{align*}
		\pfl(C,F,O) 
		&= \sum_{c\in C\setminus\{x_n,\bar x_n\}} \omega(c) d(c,O) + \underbrace{\sum_{c\in\{x_n,\bar x_n\}} \omega(c) d(c,O)}_{=4W} + \abs{O}2W \\
		&> \sum_{c\in C\setminus\{x_n,\bar x_n\}} \omega(c) d(c,O) + W + (\abs{O}+1)2W \\
		&\geq \pfl(C,F,O\cup\{x_n\}) \ .
	\end{align*}
\end{proof}

\subsection[(Phi,Psi) is a Tight Reduction]{$(\Phi,\Psi)$ is a Tight Reduction}\label{subsec:tightness}

We show that $(\Phi,\Psi)$ is a tight reduction by only considering its behaviour on reasonable solutions.
Lemma~\ref{lem:mufllocalproperty} tells us that restricted to reasonable solutions, the single-swap local search behaves on $\Phi(B,w)$ exactly the same as the flip local search behaves on $(B,w)$.
Additionally, we use the fact that once single-swap has reached a reasonable solution, it will always stay at a reasonable solution.
Formally, we need to find a set of feasible solutions $\cR$ for $(C,\omega,F,2W,d)$, such that
\begin{enumerate}
	\item $\cR$ contains all local optima.
	\item for every feasible solution $T$ of $(B,w)$, we can compute $O\in\cR$ with $T_O = T$ in polynomial time.
	\item if the transition graph $TG(C,\omega,F,2W,d)$ contains a directed path $O\leadsto O'$, with $O,O'\in\cR$ but all internal path vertices outside of $\cR$, then $TG(B,w)$ contains the edge $(T_O, T_{O'})$ or $T_O = T_{O'}$.
\end{enumerate}

Let $\cR$ be the set of all reasonable solutions.
$\cR$ contains all local optima of $(C,\omega,F,2W,d)$, by Lemma~\ref{lem:mufllocalproperty}.
The restriction of $\Psi$ to $\cR$ is bijective and we can obviously compute the inverse in polynomial time.
To prove the final property of tight reductions, we use the following result, which is a byproduct of the proof of Lemma~\ref{lem:mufllocalproperty}.

\begin{corollary}\label{cor:tg}
	If $O\in\cR$ and $O'\not\in\cR$, then $(O,O')\not\in TG(C,\omega,F,2W,d)$.
\end{corollary}

Assume $O\leadsto O'$ is a directed path in $TG(C,\omega,F,2W,d)$, with $O,O'\in\cR$ but all internal path vertices outside of $\cR$.
By Corollary~\ref{cor:tg}, this path consists of the single edge $(O,O')$.
This means that $\pfl(C,F,O) > \pfl(C,F,O')$ and thus, by Lemma~\ref{lem:muflcost}, we obtain $w(B,T_O) < w(B,T_{O'})$.
Hence, we conclude the tightness proof by observing that $(T_O,T_{O'}) \in TG(B,w)$.

\section[The K-Means Reduction]{The $K$-Means Reduction}\label{sec:km}

We complement our results by showing that we can obtain tight PLS-complete\-ness for \KMS, as well.

\begin{proposition}\label{prop:km}
	\SATF $\leq_{PLS}$ \KMS and this reduction is tight.
\end{proposition}

To prove Proposition~\ref{prop:km}, we can basically use the reduction presented in Section~\ref{subsec:construction}.
We need to change some of the constants involved in the construction to make sure that we find a set of points in $\R^D$ with the required interpoint distances.
However, the general approach stays the same and we obtain essentially the same intermediate results.
In the following, we will point out the differences in the construction of $(\Phi,\Psi)$ and indicate which proofs require adjustments.
After proving the hardness result based on the abstract definition of $C$, we show that there is indeed a point set in $\R^D$ exhibiting the required squared euclidean distances.

\subsection[Modifications to (Phi,Psi)]{Modifications to $(\Phi,\Psi)$}

As before, let $(B,w)$ be a \SAT instance over the variables $\{x_n\}_{\OneTo{n}{N}}$.
We construct an instance $(C,\omega,K)\in$ \KM.
Abstractly define the point set $C = \{x_n,\bar x_n\}_{\OneTo{n}{N}} \cup B$.
The distance function $d:C\times C\rightarrow \R$ is similar to before
\[ d(p,q) = d(q,p) = 
	\begin{cases}
		0 \cif p = q \\
		1 \cif p = x_n \wedge q = \bar x_n \\
		1+\epsilon \cif (p = x_n \vee p = \bar x_n) \wedge q = b_m \wedge p\in b_m \\
		1+c\epsilon \cif (p = x_n \vee p = \bar x_n) \wedge q = b_m \wedge \bar p\in b_m \\
		1+2\epsilon & \text{else,}
	\end{cases}
\]
where $1 < c < 2$ and 
\[ \epsilon = \frac{1}{4N+2M} \ . \]

While the distances are scaled in comparison to the MUFL reduction, the central structure remains unchanged.
The points closest to each other are literals and their negation.
Clauses are closer to literals they contain, than to the literal's negation.
All other point pairs have the same, even larger, distance to each other.

The weight function remains unchanged.
That is, the weight of a point corresponding to a clause is the SAT weight of the clause, the weight of a point corresponding to a (negated) variable is $W = M\cdot\wmax$.
Finally, we choose $K = N$. 
Like the weight function, $\Psi$ remains unchanged. 
We denote $\Phi(B,w) = (C,\omega,N)$.

\subsection{Correctness of the DKM Reduction}

Just as before, we have the problem that $\Psi$ is not injective.
However, we can again solve the problem using the previously introduced notion of reasonable solutions.
While the first condition ($\abs{O} = N$) is trivially fulfilled, we utilize the second property to ensure that $\Psi$ becomes a bijection when being restricted to reasonable solutions.
Moreover, we obtain analog results to Lemma~\ref{lem:muflcost} and \ref{lem:mufllocalproperty}.

\begin{lemma}\label{lem:dkmcost}
	If $O,O'\subset C$ are reasonable solutions for $\Phi(B,w)$, then
	\[ w(B,T_O) < w(B, T_{O'}) \Leftrightarrow \pkm(C,O) > \pkm(C,O') \ . \]
\end{lemma}

\begin{proof}
	First, we proof a claim analog to Claim~\ref{claim:appmuflcost}.
	\begin{claim}
		If $O\subset C$ is reasonable, then
		\[ \pkm(C,O) = NW + (1+\epsilon)\sum_{b_m\in B_t(X_C)} w(b_m) + (1+c\epsilon)\sum_{b_m\in B_f(X_C)} w(b_m) \ . \]
	\end{claim}

	\begin{proof}
		We have that $(T_O)_n = 1$ if $x_n\in O$ and $(T_O)_n = 0$ if $\bar x_n\in O$.
		Simply speaking, there is a one-to-one mapping of the truth assignment of a variable, to whether the corresponding positive or negative point is part of the solution.
		We obtain that $\pkm(\{x_n,\bar x_n\}_{\OneTo{n}{N}},O) = NW$, since each point corresponding to a literal is either in $O$ and has cost $0$, or its negated literal at distance $1$ is in $O$ and it has cost $W$.
		It is easy to see, by definition of the point set and $\Psi$, that the points corresponding to a clause in $B_t(X_C)$ have at least one mean at distance $1+\epsilon$ and that the points corresponding to a clause in $B_f(X_C)$ have two means at distance $1+c\epsilon$ and the rest at distance $1+2\epsilon$.
	\end{proof}
	The subsequent argument is analog to the proof presented for Lemma~\ref{lem:muflcost}.
\end{proof}

% The proof of Lemma~\ref{lem:dkmcost} is analog to the proof of Lemma~\ref{lem:muflcost}.
% The proof for Lemma~\ref{lem:dkmcost} can be obtained by substituting the modified constants into the proof of Lemma~\ref{lem:muflcost}.
Almost all of the additional work required for the DKM correctness goes into the proof of the following lemma.

\begin{lemma}\label{lem:dkmlocalproperty}
	If $O\subset C$ is a locally optimal solution for $\Phi(B,w)$, then $O$ is reasonable.
\end{lemma}

Here, we have to ensure that locally optimal solutions do not contain points corresponding to clauses.
This was not an issue in the MUFL proof, since clauses are not available for opening in that case.

Using these intermediate results we can see that this is a tight reduction following the same arguments presented in Section~\ref{subsec:tightness}

\subsection{Proof of Lemma~\ref{lem:dkmlocalproperty}}

Observe, that each point $b_m\in C$ has exactly two points at distance $1+\epsilon$ and two points at distance $1+c\epsilon$ (the points corresponding to the literals in the clause $b_m$ and their negations, respectively).
In the following, we call these four points \emph{adjacent} to $b_m$.
All the other points have distance $1+2\epsilon$ to $b_m$ and are hence strictly farther away.
Assume to the contrary that there exists an $\OneTo{n}{N}$, such that $x_n,\bar x_n\not\in O$.

\paragraph{Case 1:}
There exists an $\OneTo{m}{M}: b_m \in O$, such that $b_m = \{x_o,x_p\}$ (where one or both of these literals might be negated).
One important observation is that if we exchange $b_m$ for some other location then only its own cost and the cost of its adjacent points can increase.
All other points, which might be connected to $b_m$, are at distance $1+2\epsilon$ and can hence be connected to any other location for at most the same cost.

\paragraph{Case 1.1:}
$x_o,\bar x_o, x_p, \bar x_p \not\in  O$.
Each point adjacent to $b_m$ has weight $W$ and has distance at least $1+\epsilon$ to every other points in $P$.
Hence, we have that $\phi(\{b_m,x_o,\bar x_o, x_p, \bar x_p\}, O) \geq (4+4\epsilon)W$.
However,
\begin{align*} 
\phi(\{b_m,x_o,\bar x_o, x_p, \bar x_p\},\{x_o\}) &\leq (1+c\epsilon)\omega(b_m) + W+(2+4\epsilon)W \\
&< (4+4\epsilon)W \ ,
\end{align*}
and hence $(O\setminus\{b_m\})\cup\{x_o\}$ is in the neighbourhood of $O$ and has strictly smaller cost.

\paragraph{Case 1.2:}
$x_p\in O \vee \bar x_p\in O$ and $x_o,\bar x_o\not\in O$.
In this case, removing $b_m$ from $O$ does not affect the cost of $x_p$ and $\bar x_p$.
We obtain $\phi(\{b_m,x_o,\bar x_o\},O) \geq (2+2\epsilon)W$.
Observe, that
\[ \phi(\{b_m,x_o,\bar x_o\},(O\setminus\{b_m\})\cup\{x_o\}) \leq (1+c\epsilon)\omega(b_m) + W < 2W \ . \]

\paragraph{Case 1.3:}
$x_p\in O \vee \bar x_p\in O$ and $x_o\in O \vee \bar x_o\in O$.
Here we have that removing $b_m$ from $O$ does not affect the cost of its adjacent points at all.
However, similar to before we have $\phi(\{b_m,x_n,\bar x_n\},O) \geq (2+2\epsilon)W$.
Again, we obtain
\[ \phi(\{b_m,x_n,\bar x_n\},(C\setminus\{b_m\})\cup\{x_n\}) \leq (1+c\epsilon)\omega(b_m)+W < 2W \ . \]

\paragraph{Case 2:}
There is no $\OneTo{m}{M}$, such that $b_m\in O$.
Consequently, there is an $\OneTo{o}{N}, o\neq n: x_o,\bar x_o\in O$.
W.l.o.g. assume that $\abs{B(x_o)} < M$ (otherwise just exchange $x_o$ for $\bar x_o$ in the following argument).
Observe that
\[ \phi(B(x_o)\cup\{x_o,x_n,\bar x_n\},O) = (2+4\epsilon)W + (1+\epsilon)\sum_{b_m\in B(x_o)} \omega(b_m) \ . \]
The only points affected by removing $x_o$ from $O$ are $x_o$ and the points corresponding to clauses in $B(x_o)$.
Hence, 
\begin{align*}
\phi(C\setminus(B(x_o)\cup\{x_o,x_n,\bar x_n\}),O) &= \phi(C\setminus(B(x_o)\cup\{x_o,x_n,\bar x_n\}),O\setminus\{x_o\}) \\
&\geq \phi(C\setminus(B(x_o)\cup\{x_o,x_n,\bar x_n\}),(O\setminus\{x_o\})\cup\{x_n\}) \  .
\end{align*}
However, recall that the points in $B(x_o)$ are at distance $(1+c\epsilon)$ from $\bar x_o\in O$.
We obtain
\begin{align*}
	&\phi(B(x_o)\cup\{x_o,x_n,\bar x_n\},(O\setminus\{x_o\})\cup\{x_n\}) \\
	&\leq \phi(B(x_o)\cup\{x_o,\bar x_n\},\{\bar x_o,x_n\}) \\
	&= 2W + (1+\epsilon)\sum_{b_m\in B(x_o)} \omega(b_m)+((c-1)\epsilon)\sum_{b_m\in B(x_o)} \omega(b_m) \\
	&< 2W + (1+\epsilon)\sum_{b_m\in B(x_o)} \omega(b_m)+\epsilon W \\
	&< (2+4\epsilon)W + (1+\epsilon)\sum_{b_m\in B(x_o)} \omega(b_m) \\
	&= \phi(B(x_o)\cup\{x_o,x_n,\bar x_n\},O) \ .
\end{align*}
{\small\hfill$\square$}

\subsection[Embedding C into l22]{Embedding $C$ into $\ell_2^2$}\label{subsec:embedding}

So far, we regarded $C$ as an abstract point set, only given by fixed pairwise interpoint distances.
We show that there is an isometric embedding of $C$ into $\ell_2^2$, that is, a set of points in $\R^D$ exhibiting exactly these interpoint distances as squared Euclidean distance.
In the following, we denote by $\one_D$ the $D$-dimensional vector, where each entry is $1$, and by $\delta_{ij}$ the Kronecker delta.

\begin{theorem}[\cite{schoenberg38}]\label{thm:embedding}
	A distance matrix $M\in\R^{N\times N}$ can be embedded into $\ell_2^2$ if and only if 
	\[ \forall u\in\R^N \text{ with } u\cdot\one_N = 0 \text{ we have } u^TMu \leq 0 \ . \]
\end{theorem}

In the following let $M_C$ be the matrix corresponding to the point set $C$.
That is, we chose some ordering of the point set $C = \{c_1,\dots,c_{2N+M}\}$ and set $(M_C)_{i,j} = d(c_i,c_j)$.
Observe, that 
\[ (M_C)_{i,j} = 1 - \delta_{ij} + (d(c_i,c_j)-1)(1-\delta_{ij}) \ , \]
where the second summand is always non-negative.

\begin{lemma}
	$M_C$ can be embedded into $\ell_2^2$.
\end{lemma}

\begin{proof}
	Let $u\in\R^{2N+M}$ with $u\cdot\one_{2N+M} = 0$.
	By Theorem~\ref{thm:embedding}, it suffices to show
	\begin{align*}
		u^TM_Cu 
		&= \sum_{i,j} (M_C)_{i,j} u_i u_j \\
		&= \sum_{i,j} u_i u_j - \sum_{i,j} u_i u_j \delta_{ij} + \sum_{i,j} (d(c_i,c_j)-1) u_i u_j (1-\delta_{ij}) \\
		&= (\underbrace{\sum_i u_i}_{=0})^2 - \sum_i u_i^2 + \sum_{i,j} (d(c_i,c_j)-1) u_i u_j (1-\delta_{ij}) \\
		&\leq -\norm{u}^2 + \sum_{i,j} (d(c_i,c_j)-1) \abs{u_i} \abs{u_j} (1-\delta_{ij}) \\
		&\leq -\norm{u}^2 + 2\epsilon \sum_{i,j} \abs{u_i} \abs{u_j} \\
		&= -\norm{u}^2 + 2\epsilon (\sum_i \abs{u_i})^2 \\
		&\leq -\norm{u}^2 + 2\epsilon (2N+M) \norm{u}^2 = 0 \ ,
	\end{align*}
	where the second to last inequality holds by Cauchy-Schwarz.
\end{proof} 

\begin{theorem}[\cite{torgerson52}]
	If $M$ is a distance matrix embeddable into $\ell_2^2$, then there is a polynomial-time algorithm that computes a matrix $P$ whose rows form a set $\{p_n\}_{\OneTo{n}{N}}$ with $M_{i,j} = \norm{p_i-p_j}^2$.
\end{theorem}
\section[Application to Fuzzy K-Means]{Application to Fuzzy $K$-Means}

Clustering problems, such as DKM, appear in many applications of machine learning and data mining and are closely related to facility location problems.
Problems like DKM, where each point is assigned to a single location, are sometimes called \emph{hard clustering} problems.
If we allow clusters to overlap, so that a point can be assigned to multiple locations, then we speak of a \emph{soft} clustering.
One popular soft clustering generalization of the $K$-means problem is the \emph{fuzzy} $K$-means problem.
In addition to the $K$ location vectors, the fuzzy $K$-means problem seeks for a set of \emph{memberships} assigning some fraction of each point to each location.
Formally, the goal of the \FKM (DFKM) problem, given $C\subset\R^D$ and $K\in\N$, is to find $O\subset C$ and $r: C\times O \rightarrow \R_{\geq 0}$ minimizing
\[ \pfkm(C,O,r) = \sum_{c\in C} \sum_{o\in O} r(c,o)^2 \norm{c-o}^2 , \ \]
subject to $\forall c\in C: \sum_{o\in O} r(c,o) = 1$.

For each fixed set of locations $O\subset\R^D$ we can compute a membership function $r$ minimizing $\phi_{FKM}(C,O,r)$ in polynomial time \cite{dunn73}.
Hence, we obtain that \FKMS is a PLS problem.
So far, there are no known results on the quality of the solutions produced by the standard local search algorithm of \FKMS.
The reduction presented for \KMS can be generalized to apply to \FKMS, as well.

\begin{theorem}\label{thm:fuzzy}
	\FKMS is tightly PLS-complete.
\end{theorem}

We prove Theorem~\ref{thm:fuzzy} by presenting a tight PLS reduction of \FKMS to \NAE (PNAESAT)/Flip, a variant of the previously introduced \SATF.

\subsection{Preliminaries}

An instance of PNAESAT is a weighted Boolean formula in conjunctive normal form, where each clause consists of $2$ \emph{positive} literals, and the cost of a truth assignment is given by the sum over the weights of all clauses whose two variables differ in their truth value.
As before, in the PLS problem \NAEF, the neighbourhood of an assignment is given by all assignments obtained by changing the truth value of a single variable.
For each clause set $B$ and truth assignment $T$ we denote the NAE cost of $T$ with respect to $B$ by $w_{NAE}(B,T)$. %, and

\begin{theorem}[\cite{schaeffer91}]
	\NAEF is tightly PLS-complete.
\end{theorem}

\begin{proposition}\label{prop:fkm}
	\NAEF $\leq_{PLS}$ \FKMS and this reduction is tight.
\end{proposition}

The following proof of Proposition~\ref{prop:fkm} will basically use the same reduction previously presented for \KMS.
However, first we need to modify the given PNAESAT instance slightly, and change some of the constants involved.

\subsection[Modification to (B,w)]{Modification to $(B,w)$}

As before, let $(B = \{b_m\}_{\OneTo{m}{M}},w)$ be a PNAESAT instance over the variables $\{x_n\}_{\OneTo{n}{N}}$.
From $B$ we construct a new set of clauses $B'$ which will be the input to our reduction.
For each clause $b_m = \{x_o,x_p\} \in B$ $B'$ contains the two clauses 
\[ b_m^1 = \{x_o,x_p\} \text{ and } b_m^2 = \{\bar x_o, \bar x_p\} \ . \]
We define our reduction function based on $(B',w)$.
We still denote $\wmax = \max_{\OneTo{m}{M}}\{w(b_m)\}$, but redefine $M = \abs{B'}$.
% Additionally, we assume that $N > 2$.

\subsection[Construction of Phi and Psi]{Construction of $\Phi$ and $\Psi$}

Given $(B',w)$, we construct an instance $(C',\omega,K)\in$ \FKM.
As before, abstractly define the point set $C' = \{x_n,\bar x_n\}_{\OneTo{n}{N}} \cup B$.
The distance function $d:C\times C\rightarrow \R$ is again
\[ d(p,q) = d(q,p) = 
	\begin{cases}
		0 \cif p = q \\
		1 \cif p = x_n \wedge q = \bar x_n \\
		1+\epsilon \cif (p = x_n \vee p = \bar x_n) \wedge q = b_m \wedge p\in b_m \\
		1+c\epsilon \cif (p = x_n \vee p = \bar x_n) \wedge q = b_m \wedge \bar p\in b_m \\
		1+2\epsilon & \text{else,}
	\end{cases}
\]
where $1 < c < 2$.
However, we set 
\[ \epsilon = \min\left\{  \frac{1}{4N+2M},\frac{M-1}{9N^2M} \right\} \ . \]

Let the weight function $\omega$ be defined as before
\[
	\omega(c) = 
	\begin{cases}
		w(b_m) \cif c = b_m^i \\
		W & \text{else,}
	\end{cases}
\]
however we choose $W = 4N^2\cdot M\cdot \wmax$.
As before, let $K=N$.
The function $\Psi$ remains unchanged. 

\subsection[Properties of (C',omega,N)]{Properties of $(C',\omega,N)$}

Recall the following important properties of the fuzzy $K$-means objective function.

\begin{lemma}[\cite{dunn73}]
	Let $C\subset\R^D$ and fix any set of means $O\subset\R^D$.
	For each $c\in C$ and $o\in O$ we have that the optimal membership of $c$ to the mean $o$ is given by
	\[ r(c,o) = \frac{\norm{c-o}^{-2}}{\sum_{o'\in O} \norm{c-o'}^{-2}} \ . \]
	Substituting for all optimal memberships we obtain
	\[ \pfkm(C,O) = \sum_{c\in C\setminus O} \frac{\omega(c)}{\sum_{o\in O}\norm{c-o}^{-2}} \ . \]
\end{lemma}

\begin{lemma}\label{lem:lowerboundr}
	For all $O\subset C'$ with $\abs{O} = N$ we have 
	\[ \forall c\in C'\setminus O \; \forall o\in O: r(c,o) > \frac{1}{2N} \ . \]
\end{lemma}

\begin{proof}
	By definition of the interpoint distances in $C'$ we obtain
	\[ r(c,o) = \frac{\norm{c-o}^{-2}}{\sum_{o'\in O}\norm{c-o'}^{-2}} \geq \frac{(1+2\epsilon)^{-1}}{\abs{O}} > \frac{1}{2N} \ . \]
\end{proof}

\subsection[(Phi,Psi) is a PLS-Reduction]{$(\Phi,\Psi)$ is a PLS-Reduction}

\begin{lemma}\label{lem:fkmeanscost}
	If $O\subset C'$ is reasonable, then 
	\[ \pfkm(C',O) = \frac{N+N2\epsilon}{N+2\epsilon} + 2\Gamma_2\sum_{b_m\in B_t(T_O)}w(b_m) + (\Gamma_1+\Gamma_3)\sum_{b_m\in B_f(T_O)}w(b_m) \ , \]
	where
	\[ \Gamma_1 = \frac{1}{\frac{N-2}{1+2\epsilon}+\frac{2}{1+\epsilon}} \text{, }\quad \Gamma_2 = \frac{1}{\frac{N-2}{1+2\epsilon}+\frac{1}{1+\epsilon}+\frac{1}{1+c\epsilon}} \text{, and }\; \Gamma_3 = \frac{1}{\frac{N-2}{1+2\epsilon}+\frac{2}{1+c\epsilon}} \ . \]
\end{lemma}

\begin{proof}
	We have that $(T_O)_n = 1$ if $x_n\in C$ and $(T_O)_n = 0$ if $\bar x_n\in O$.
	Simply speaking, there is a one-to-one mapping of the truth assignment of a variable, to whether the corresponding positive or negative point is part of the solution.
	We obtain that $\pfkm(\{x_n,\bar x_n\}_{\OneTo{n}{N}},O) = \frac{N+N2\epsilon}{N+2\epsilon}$, since each point corresponding to a literal is either in $O$ and has cost $0$, or it has its negated literal at distance $1$ and $N-1$ means at distance $1+2\epsilon$ and hence has cost
	\[ \frac{1}{\frac{N-1}{1+2\epsilon}+1} = \frac{{1+2\epsilon}}{N+2\epsilon} \ . \]
	Recall, that a clause of PNAESAT is satisfied if one variable evaluates to true, and one to false.
	Hence, the points $b_m^1$ and $b_m^2$ corresponding to a clause in $b_m\in B_t(T_O)$ have one mean at distance $1+\epsilon$, one mean at distance $1+c\epsilon$ and $N-2$ means at distance $1+2\epsilon$.
	We obtain, that $\forall b_m\in B_t(T_O): \pfkm(\{b_m^1,b_m^2\},O) = 2w(b_m)/\left(\frac{N-2}{1+2\epsilon}+\frac{1}{1+\epsilon}+\frac{1}{1+c\epsilon}\right)$.
	Conversely, the points $b_m^1$ and $b_m^2$ corresponding to a clause in $B_f(T_O)$ split in two different groups.
	One of them has two means at distance $1+\epsilon$ and the other has two means at distance $1+c\epsilon$, while both have $N-2$ means at distance $1+2\epsilon$.
	We obtain, that $\forall b_m\in B_f(T_O): \pfkm(\{b_m^1,b_m^2\},O) = w(b_m)/\left(\frac{N-2}{1+2\epsilon}+\frac{2}{1+\epsilon}\right)+w(b_m)/\left(\frac{N-2}{1+2\epsilon}+\frac{2}{1+c\epsilon}\right)$.
\end{proof}

\begin{lemma}
	If $O,O'\subset C'$ are reasonable, then
	\[ w_{NAE}(B,T_O) < w_{NAE}(B,T_{O'}) \Leftrightarrow \pfkm(C',O) > \pfkm(C',O') \ . \]
\end{lemma}

\begin{proof}
	Observe that 
	\begin{align}
		w(B,T_O) < w(B,T_{O'}) \Leftrightarrow \sum_{b_m\in B_{tf}} w(b_m) < \sum_{b_m \in B_{ft}} w(b_m) \ , \label{eq:localoptf}
	\end{align}
	and
	\begin{align}
	\begin{split}
		\Gamma_1+\Gamma_3-2\Gamma_2
		&= \frac{1}{\frac{N-2}{1+2\epsilon}+\frac{2}{1+\epsilon}} + \frac{1}{\frac{N-2}{1+2\epsilon}+\frac{2}{1+c\epsilon}}-\frac{2}{\frac{N-2}{1+2\epsilon}+\frac{1}{1+\epsilon}+\frac{1}{1+c\epsilon}} \\
		&= \frac{2\left(\frac{1}{1+c\epsilon}-\frac{1}{1+\epsilon}\right)^2}{\left(\frac{N-2}{1+2\epsilon}+\frac{2}{1+\epsilon}\right)\left(\frac{N-2}{1+2\epsilon}+\frac{2}{1+c\epsilon}\right)\left(\frac{N-2}{1+2\epsilon}+\frac{1}{1+\epsilon}+\frac{1}{1+c\epsilon}\right)} > 0 \ .
	\end{split}\label{eq:gamma}
	\end{align}
	Hence, using Lemma~\ref{lem:fkmeanscost} we obtain
	\begin{align*}
		\pfkm(C',T_{O'}) 
		= \frac{N+N2\epsilon}{N+2\epsilon} + 2\Gamma_2&\sum_{b_m\in B_t(T_{O'})}w(b_m) + \\
		(\Gamma_1+\Gamma_3)&\sum_{b_m\in B_f(T_{O'})}w(b_m) \\
		= \frac{N+N2\epsilon}{N+2\epsilon} + 2\Gamma_2&\sum_{b_m\in B_{tt}}w(b_m) + 2\Gamma_2\sum_{b_m\in B_{ft}}w(b_m) +\\
		(\Gamma_1+\Gamma_3)&\sum_{b_m\in B_{tf}}w(b_m) + (\Gamma_1+\Gamma_3)\sum_{b_m\in B_{ff}}w(b_m)\\
		\overset{\eqref{eq:localoptf},\eqref{eq:gamma}}{<} \frac{N+N2\epsilon}{N+2\epsilon} + 2\Gamma_2&\sum_{b_m\in B_{tt}}w(b_m) + (\Gamma_1+\Gamma_2)\sum_{b_m\in B_{ft}}w(b_m) +\\
		2\Gamma_2&\sum_{b_m\in B_{tf}}w(b_m) + (\Gamma_1+\Gamma_3)\sum_{b_m\in B_{ff}}w(b_m)\\
		= \frac{N+N2\epsilon}{N+2\epsilon} + 2\Gamma_2&\sum_{b_m\in B_t(T_O)}w(b_m) + \\
		(\Gamma_1+\Gamma_3)&\sum_{b_m\in B_f(T_O)}w(b_m) \\
		= \pfkm(C',T_O) \ .
	\end{align*}
\end{proof}

\begin{lemma}
	If $O$ is locally optimal for $(C',\omega,N)$, then $O$ is reasonable
\end{lemma}

\begin{proof}
	Still, each $b_m^i\in C'$ has four adjacent points.
	Assume to the contrary that there is an $\OneTo{n}{N}$, such that $x_n,\bar x_n\not\in O$.
	
	\paragraph{Case 1:}
	$b_m^i\in O$, $b_m^i = \{x_o,x_p\}$.
	Similar to before, we observe that exchanging $b_m^i$ for some other mean can only increase the cost of its adjacent points.
	All other points have $b_m^i$ as a mean at distance $1+2\epsilon$, thus an exchange can not move the mean farther away, hence their cost can not increase.
	
	\paragraph{Case 1.1:}
	$x_o,\bar x_o, x_p,\bar x_p\not\in O$.
	Let $r$ be the optimal memberships of $C'$ with respect to $O$, $\tr$ be the optimal memberships with respect to $\tO = (O\setminus\{b_m^i\})\cup\{x_o\}$, denote 
	\[ \textstyle O(x) = \sum_{o\in O\setminus\{b_m^i\}} r(x,o)^2 \norm{x-o}^2 \quad\text{and}\quad \tO(x) = \sum_{o\in O\setminus\{b_m^i\}} \tr(x,o)^2 \norm{x-o}^2 \ . \]
	\begin{align*}
		&\pfkm(C',O) - \pfkm(C',\tO) \\
		\geq &\pfkm(\{b_m^i,x_o,\bar x_o,x_p,\bar x_p\},O) - \pfkm(\{b_m^i,x_o,\bar x_o,x_p,\bar x_p\},\tO) \\
		= &\sum_{x\in\{x_o,\bar x_o,x_p,\bar x_p\}} \omega(x)(r(x,b_m^i)^2\norm{x-b_m^i}^2 + O(x)) - \\
		&\sum_{x\in\{b_m^i,\bar x_o,x_p,\bar x_p\}} \omega(x)(\tr(x,x_o)^2 \norm{x-x_o}^2 + \tO(x))
	\end{align*}
	Recall, that we can only increase the cost by substituting for non-optimal memberships.
	\begin{align*}
		\geq &\sum_{x\in\{x_o,\bar x_o,x_p,\bar x_p\}} \omega(x)(r(x,b_m^i)^2 \norm{x-b_m^i}^2 + O(x)) - \\
		&\sum_{x\in\{b_m^i,\bar x_o,x_p,\bar x_p\}} \omega(x)(r(x,b_m^i)^2 \norm{x-x_o}^2 + O(x)) \\
		= &\omega(x_o)(r(x_o,b_m^i)^2\norm{x_o-b_m^i}^2 + O(x_o)) - \\
		&\omega(b_m^i)\underbrace{r(b_m^i,b_m^i)^2}_{=1}(\norm{b_m^i-x_o}^2 +\underbrace{O(b_m^i)}_{=0}) + \\
		 &\sum_{x\in\{\bar x_o, x_p, \bar x_p\}} \omega(x)r(x,b_m^i)^2 (\norm{x-b_m^i}^2 - \norm{x-x_o}^2) \\ 
		 \geq &(r(x_o,b_m^i)^2 W - \omega(b_m^i))\norm{x_o-b_m^i}^2 + Wr(\bar x_o,b_m^i)^2c\epsilon + \\
		 &W(\epsilon-2\epsilon)(r(x_p,b_m^i)^2 + r(\bar x_p,b_m^i)^2)
	\end{align*}
	Using Lemma~\ref{lem:lowerboundr} we obtain
	\begin{align*}
		\geq &(\frac{1}{4N^2}W - \omega(b_m^i))\norm{x_o-b_m^i}^2 - 2W\epsilon = \\
		&(M\wmax - \omega(b_m^i))\norm{x_o-b_m^i}^2 - 2W\epsilon \geq (M-1)\wmax(1+\epsilon) - 2W\epsilon \ .
	\end{align*}
	Since $0 < \epsilon \leq \frac{M-1}{9N^2M}$ we finally obtain
	\[ \geq (M-1)\wmax - 8N^2M\wmax\frac{M-1}{9N^2M} > 0 \ . \]
	
	\paragraph{Case 1.2:}
	$\bar x_p\in O$ and $x_o,\bar x_o\not\in O$.
	\begin{align*}
		&\pfkm(C',O) - \pfkm(C',\tO) \\
		\geq &\pfkm(\{b_m^i,x_o,\bar x_o,x_p,\bar x_p\},O) - \pfkm(\{b_m^i,x_o,\bar x_o,x_p,\bar x_p\},\tO) \\
		\geq &\sum_{x\in\{x_o,\bar x_o, x_p\}} \omega(x)(r(x,b_m^i)^2 \norm{x-b_m^i}^2 + O(x)) - \\
		&\sum_{x\in\{b_m^i,\bar x_o, x_p,\}} \omega(x)(r(x,b_m^i)^2 \norm{x-x_o}^2 + O(x)) \\
		= &W(r(x_o,b_m)^2\norm{x_o-b_m^i}^2+O(x_o))- \\
		&\omega(b_m^i)\underbrace{r(b_m^i,b_m^i)^2}_{=1}(\norm{b_m^i- x_o}^2+\underbrace{O(b_m^i)}_{=0}) + \\
		&\sum_{x\in\{\bar x_o, x_p\}}\omega(x)r(x,b_m^i)^2(\norm{x-b_m^i}^2-\norm{x-x_o}^2) \\
		\geq &(\frac{1}{4N^2}W-\omega(b_m^i))\norm{x_o-b_m^i}^2 + Wr(\bar x_o,b_m^i)^2\epsilon + Wr(x_p,b_m^i)^2(\epsilon - 2\epsilon)\\
		\geq &(M\wmax - \omega(b_m^i))\norm{x_o-b_m^i}^2 - \epsilon W \\
		\geq &(M-1)\wmax\norm{x_o-b_m^i}^2 - \epsilon W \geq (M-1)\wmax(1+\epsilon) - \epsilon W
	\end{align*}
	Since $0 < \epsilon \leq \frac{M-1}{9N^2M}$ we finally obtain
	
	\[ \geq (M-1)\wmax - 4N^2M\wmax\frac{M-1}{9N^2M} > 0 \ . \]
	
	\paragraph{Case 1.3:}
	$\bar x_p\in O$, $\bar x_o\in O$.
	\begin{align*}
		&\pfkm(C',O) - \pfkm(C',(C\setminus \{b_m^i\})\cup\{x_n\}) \\
		\geq &\pfkm(\{b_m^i,x_o,\bar x_o,x_p,\bar x_p,x_n,\bar x_n\},O) - \\
		&\pfkm(\{b_m^i,x_o,\bar x_o,x_p,\bar x_p,x_n,\bar x_n\},(C\setminus \{b_m\})\cup\{x_n\}) \\
		\geq &\sum_{x\in\{x_o, x_p,x_n,\bar x_n\}} \omega(x)(r(x,b_m^i)^2 \norm{x-b_m^i}^2 + O(x)) - \\
		&\sum_{x\in\{b_m^i,x_o, x_p,\bar x_n\}} \omega(x)(r(x,b_m^i)^2 \norm{x-x_n}^2 + O(x)) \\
		= &W(r(x_n,b_m^i)^2\norm{x_n-b_m^i}^2+O(x_n)) - \\
		&\omega(b_m^i)\underbrace{r(b_m^i,b_m^i)^2}_{=1}(\norm{b_m^i-x_n}^2+\underbrace{O(b_m^i)}_{=0}) + \\
		&\sum_{x\in\{x_o,x_p,\bar x_n\}}\omega(x)r(x,b_m^i)^2(\norm{x-b_m^i}^2-\norm{x-x_n}^2) \\
		\geq &(\frac{1}{4N^2}W-\omega(b_m^i))(1+2\epsilon) + Wr(\bar x_n,b_m^i)^2 2\epsilon + \\
		&W(r(x_o,b_m^i)^2+r(x_p,b_m^i)^2)(\epsilon - 2\epsilon) \\
		\geq &(M\wmax - \omega(b_m^i))(1+2\epsilon) - 2\epsilon W \geq (M-1)\wmax(1+2\epsilon) - 2\epsilon W
	\end{align*}
	Since $0 < \epsilon \leq \frac{M-1}{9N^2M}$ we finally obtain
	
	\[ \geq (M-1)\wmax - 8N^2M\wmax\frac{M-1}{9N^2M} > 0 \ . \]
	
	\paragraph{Case 2:}
	There is no $\OneTo{m}{M}, \OneTo{i}{3}: b_m^i\in O$, hence $\exists o, o\neq n: x_o,\bar x_o\in O$.
	$B(x_o) = \{b_m^i\in B \;|\; x_o\in b_m^i \vee \bar x_o\in b_m^i\}$, $\abs{B(x_o)} < M$ (otherwise just exchange $x_o$ for $\bar x_o$ in the following argument).
	Observe that
	\begin{align*}
		\pfkm(B(x_o)\cup\{x_o,x_n,\bar x_n\},O) = &\underbrace{\frac{W(2+4\epsilon)}{N}}_{\pfkm(\{x_n,\bar x_n\},O)} + \\
		&\sum_{b_m^i\in B(x_o)} \omega(b_m^i) \sum_{c\in C} r(b_m^i,c)^2\norm{b_m^i - c}^2 \ . 
	\end{align*}
	
	The only points whose cost can increase by removing $x_o$ from $C$ are the points in $B(x_o)$ and $x_o$ itself.
	All points in $C\setminus (B(x_o)\cup\{x_o\})$ are at distance $1+2\epsilon$ anyways, and hence exchanging $x_o$ for some other mean can not increase their cost.
	We obtain
	\begin{align*}
		&\pfkm(B(x_o)\cup\{x_o,x_n,\bar x_n\},(C\setminus\{x_o\})\cup\{x_n\}) \\
		\leq &\underbrace{\frac{W(2+4\epsilon)}{N+2\epsilon}}_{\pfkm(\{x_o,\bar x_n\},(C\setminus\{x_o\})\cup\{x_n\})} + \\
		&\sum_{b_m^i\in B(x_o)} \omega(b_m^i) \left( r(b_m^i,x_o)^2 \underbrace{\norm{b_m^i-x_n}^2}_{=1+2\epsilon \leq 1+2\epsilon} + \sum_{c\in C\setminus\{x_o\}} r(b_m^i,c)^2\norm{b_m^i - c}^2\right) \\
		\leq &\frac{W(2+4\epsilon)}{N+2\epsilon} + \\
		&\sum_{b_m^i\in B(x_o)} \omega(b_m^i) \left( r(b_m^i,x_o)^2 (\norm{b_m^i-x_o}^2+\epsilon) + \sum_{c\in C\setminus\{x_o\}} r(b_m^i,c)^2\norm{b_m^i - c}^2\right) \\
		= &\frac{W(2+4\epsilon)}{N+2\epsilon} + \epsilon\sum_{b_m^i\in B(x_o)} \omega(b_m^i)r(b_m^i,x_o)^2 + \\
		&\sum_{b_m^i\in B(x_o)} \omega(b_m^i) \sum_{c\in C} r(b_m^i,c)^2\norm{b_m^i - c}^2 \\
		< &\frac{W(2+4\epsilon)}{N+2\epsilon} + \epsilon M\wmax + \sum_{b_m^i\in B(x_o)} \omega(b_m^i) \sum_{c\in C} r(b_m^i,c)^2\norm{b_m^i - c}^2 \\
		= &\frac{W(2+4\epsilon)}{N+2\epsilon} + \epsilon \frac{W}{4N^2} + \sum_{b_m^i\in B(x_o)} \omega(b_m^i) \sum_{c\in C} r(b_m^i,c)^2\norm{b_m^i - c}^2 \\
		< &\frac{W(2+4\epsilon)}{N+2\epsilon} + \frac{\frac{\epsilon}{N}W}{N+2\epsilon} + \sum_{b_m^i\in B(x_o)} \omega(b_m^i) \sum_{c\in C} r(b_m^i,c)^2\norm{b_m^i - c}^2 \\
		< &\frac{W(2+4\epsilon)+\frac{2\epsilon(2+4\epsilon)}{N}W}{N+2\epsilon} + \sum_{b_m^i\in B(x_o)} \omega(b_m^i) \sum_{c\in C} r(b_m^i,c)^2\norm{b_m^i - c}^2 \\
		= &\frac{W(2+4\epsilon)(1+\frac{2\epsilon}{N})}{N(1+\frac{2\epsilon}{N})} + \sum_{b_m^i\in B(x_o)} \omega(b_m^i) \sum_{c\in C} r(b_m^i,c)^2\norm{b_m^i - c}^2 \\
		= &\frac{W(2+4\epsilon)}{N} + \sum_{b_m^i\in B(x_o)} \omega(b_m^i) \sum_{c\in C} r(b_m^i,c)^2\norm{b_m^i - c}^2 \\
		= &\pfkm(B(x_o)\cup\{x_o,x_n,\bar x_n\},O)
	\end{align*}
\end{proof}

\begin{corollary}
	If $O$ is locally optimal for $(C',\omega,N)$, then $T_O$ is locally optimal for $(B,w)$.
\end{corollary}

Proving tightness of $(\Phi,\Psi)$ and embedding of $C'$ into $\ell_2^2$ follows the same line or arguments presented for \KM in Section~\ref{sec:km}.
\section{Open Problems}

In this work, we explore the local search complexity of the single-swap heuristic for MUFL and DKM.
While we prove that the problem is tightly PLS-complete in general, our reduction requires arbitrarily many dimensions, number of clusters and a non trivial weight function on the clients.
% One of the first follow-up question is if we can reduce either the number of clusters $K$ or the dimension $D$ down to a constant.
One of the first follow-up question is if we can reduce the number of dimensions $D$ down to a constant.
Moreover, it is interesting to examine whether we can obtain our results for unweighted variants of these problems.
The fact that the $K$-means method has exponential worst-case runtime even for unweighted point sets with $D=2$ indicates that this might be possible.
A potential approach to reduce the number of dimension is e.g. to embed our abstract point set using different techniques than the one presented here, since this is the only point in the proof that requires high dimensionality.
% Using a reduction as it is presented here would require to somehow encode the weights of clauses into interpoint distances.
% It is not clear how to ensure that these distances form a metric (the MUFL case) or are embeddable into $\ell_2^2$ (the DKM case).
% Moreover, our result only holds for \FKM with $m = 2$.
% It would be interesting to examine if it is possible to extend this result to arbitrary $m$.

% Furthermore, it is interesting to analyze \FKM algorithmically.
% Recall Kanungo's result, that the search heuristic yields an $\cO(1)$ approximation of the $K$-means problem \cite{kanungo04}.
% We believe it possible, that this result (although most likely with a larger constant) transfers to fuzzy $K$-means.

The major open result is still the conjecture of Roughgarden and Wang, that computing a local minimum of the $K$-means algorithm is a PLS-hard problem \cite{roughgarden16}.

\bibliographystyle{alpha}
\bibliography{ref}

\end{document}